\documentclass[10pt,reqno]{amsart}
     \makeatletter
     \def\section{\@startsection{section}{1}%
     \z@{.7\linespacing\@plus\linespacing}{.5\linespacing}%
     {\bfseries
     \centering
     }}
     \def\@secnumfont{\bfseries}
     \makeatother
\newtheorem{theorem}{Theorem}[section]

\theoremstyle{definition}

\theoremstyle{remark}

\numberwithin{equation}{section}
\usepackage{amsmath}
\usepackage{amssymb}
\usepackage{amsfonts}
\usepackage{graphicx}
\usepackage{subfigure}
\begin{document}

\title[Steady filtration of Peng-Robinson gas in a porous medium]{Steady filtration of Peng-Robinson gases in a porous medium}

\author{Valentin Lychagin}
\address{Valentin Lychagin: V.A. Trapeznikov Institute of Control Sciences, Russian Academy of Sciences, 65 Profsoyuznaya Str., 117997 Moscow, Russia, Department of Mathematics, The Arctic University of Norway, Postboks 6050, Langnes 9037, Tromso, Norway}
\email{valentin.lychagin@uit.no}

\author{Mikhail Roop}
\address{Mikhail Roop: V.A. Trapeznikov Institute of Control Sciences, Russian Academy of Sciences, 65 Profsoyuznaya Str., 117997 Moscow, Russia, Department of Physics, Lomonosov Moscow State University, Leninskie Gory, 119991 Moscow, Russia}
\email{mihail\underline{ }roop@mail.ru}



\begin{abstract}
Filtration of real gases described by Peng-Robinson equations of state in 3-dimensional space is studied. Thermodynamic states are considered as either Legendrian submanifolds in contact space, or Lagrangian submanifolds in symplectic space. The correspondence between singularities of their projection on the plane of intensives and phase transitions is shown, and coexistence curves in various coordinates are constructed. A method of finding explicit solutions of the Dirichlet boundary problem is provided and the case of a number of sources is discussed in details. The domains corresponding to different phases are shown.
\end{abstract}

\maketitle

\section{Introduction}
A system of equations describing a steady filtration in a 3-dimensional porous medium consists of \cite{Lib,Mus, Sch}
\begin{itemize}
\item the Darcy law
\begin{equation}
\label{Darcy}\mathbf{u}=-\frac{k}{\mu}\nabla p,
\end{equation}
where $\mathbf{u}(x)=(u_{1},u_{2},u_{3})$ is the velocity field, $p(x)$ is the pressure, $x\in\mathbb{R}^{3}(x_{1},x_{2},x_{3})$, $k=\|k_{ij}\|$ is the permeability tensor, depending on the medium as well as the viscosity $\mu$. The Darcy law in form (\ref{Darcy}) is valid for one-component filtration, i.e. medium consists of only one sort of fluid or gas. Since we consider homogeneous medium, we put $k_{ij}=k\delta_{ij}$, where $\delta_{ij}$ is the Kronecker symbol.
\item the continuity equation
\begin{equation}
\label{Cont}
\mathrm{div}(\rho\mathbf{u})=0,
\end{equation}
where $\rho(x)$ is the density. Equation (\ref{Cont}) is responsible for the mass conservation law.
\end{itemize}
In addition to (\ref{Darcy}) and (\ref{Cont}) we assume that the specific entropy of the gas $\sigma(x)$ is constant along the trajectories of the velocity field $\mathbf{u}$:
\begin{equation}
\label{adiab}
(\mathbf{u},\nabla\sigma)=0.
\end{equation}

Note that system (\ref{Darcy})-(\ref{adiab}) is underdetermined. To make it complete, we need additional relations representing thermodynamic properties of the medium, i.e. equations of state. Filtration of ideal gases was investigated in \cite{LychGSA}. In this paper, we use one of the most popular in petroleum industry model of real gases, the Peng-Robinson model \cite{PR}. One of the most important properties of such gases is phase transitions. We obtain \textit{coexistence curves}, i.e. sets of points where phase transition occurs, for such gases in the space of thermodynamic variables and having solutions of the Dirichlet problem for (\ref{Darcy})-(\ref{adiab}) extended by Peng-Robinson equations of state, we can move these curves onto the space $\mathbb{R}^{3}(x_{1},x_{2},x_{3})$ and find the domains where phase transitions occur. Similar results have already been obtained by authors in \cite{GLRT} for the Navier-Stokes flows, in \cite{LRNon} for the van der Waals gases filtration and in \cite{LRJGP} for the Redlich-Kwong gases filtration.
\section{Thermodynamic state}
Here, we briefly recall (for details see \cite{LY, LRNon}) necessary constructions describing thermodynamic states by means of contact and symplectic geometry \cite{KLR}. Let us introduce a contact manifold $(\mathbb{R}^{5},\theta)$ with coordinates $(p,T,e,v,\sigma)$ standing for the pressure, the temperature, the specific energy, the specific volume $v=\rho^{-1}$ and the specific entropy respectively. The structure form $\theta$ represents the first law of thermodynamics:
\begin{equation*}\theta=d\sigma-T^{-1}de-T^{-1}pdv.\end{equation*}
In our consideration, the thermodynamic state is a 2-dimensional Legendrian submanifold $\widehat L\subset(\mathbb{R}^{5},\theta)$, i.e. maximal integral manifold of the form $\theta$:
\begin{equation*}\theta|_{\widehat L}=0,\end{equation*}
which means that the first law of thermodynamics holds on $\widehat L$.

Since thermodynamic state is defined by interplay of measurable quantities, it is reasonable to eliminate the specific entropy, which can be done by projection $\pi\colon\mathbb{R}^{5}\to\mathbb{R}^{4}$, $\pi(p,T,e,v,\sigma)=(p,T,e,v)$. The restriction of this projection on $\widehat L$ gives an immersed Lagrangian submanifold $L\subset(\mathbb{R}^{4},\Omega)$, such as $\Omega|_{L}=0$, where $\Omega$ is a symplectic form
\begin{equation*}\Omega=-d\theta=d(T^{-1})\wedge de+d(pT^{-1})\wedge dv.\end{equation*}
Thus, in $(\mathbb{R}^{4},\Omega)$ the Lagrangian submanifold $L$ is given by \textit{equations of state}:
\begin{equation*}L=\left\{f(p,T,e,v)=0,\hspace{1mm}g(p,T,e,v)=0\right\}.\end{equation*}
Then, the condition for $L$ to be Lagrangian is that the Poisson bracket $[f,g]$ between functions $f$ and $g$ with respect to structure form $\Omega$
\begin{equation*}[f,g]\hspace{1mm}\Omega\wedge\Omega=df\wedge dg\wedge\Omega\end{equation*}
vanishes on the surface $L$:
\begin{equation}\label{compat}[f,g]=0\text{ on }L.\end{equation}

We will consider gases with thermodynamic states obeying equations of state in the form
\begin{equation*}f(p,T,e,v)=p-A(v,T),\quad g(p,T,e,v)=e-B(v,T).\end{equation*}
The first equation is called \textit{thermic} equation of state, while the second one is called \textit{caloric} equation of state. Usually the thermic equation of state is derived from experiments, but the caloric one remains unknown. Using the compatibility condition (\ref{compat}) one can get the caloric equation. Moreover, the following theorem is valid \cite{LRNon}:
\begin{theorem}
The Legendrian submanifold $\widehat L$ is given by the Massieu-Plank potential $\phi(v,T)$:
\begin{equation}\label{LegMani}p=RT\phi_{v},\quad e=RT^{2}\phi_{T},\quad\sigma=R(\phi+T\phi_{T}),\end{equation}
where $R$ is the universal gas constant.
\end{theorem}
Thus, by a system of filtration equations we shall mean equations (\ref{Darcy})-(\ref{adiab}) extended by the Legendrian surface $\widehat L$.

In general, not all the points on $L$ are applicable. This means that not all possible combinations of thermodynamic variables satisfy the condition of thermodynamic stability. The set of applicable points is defined by means of the quadratic differential form $\kappa$ on $L$ \cite{LY}:
\begin{equation*}\kappa=d(T^{-1})\cdot de+d(pT^{-1})\cdot dv,\end{equation*}
which in terms of Massieu-Plank potential takes the following form:
\begin{equation*}R^{-1}\kappa=-(\phi_{TT}+2T^{-1}\phi_{T})dT\cdot dT+\phi_{vv}dv\cdot dv.\end{equation*}
\begin{theorem}
Applicable states on $L$ correspond to the set of points where the differential quadratic form $\kappa$ is negative and are given by inequalities:
\begin{equation*}\phi_{TT}+2T^{-1}\phi_{T}>0,\quad\phi_{vv}<0,\end{equation*}
or, equivalently,
\begin{equation*}e_{T}>0,\quad p_{v}<0.\end{equation*}
\end{theorem}
Consequently, we may have two types of singular submanifolds on $L$:
\begin{itemize}
\item singular submanifold $\Sigma_{1}\subset L$, where the differential form $de\wedge dv$ degenerates:
\begin{equation*}\Sigma_{1}=\left\{\phi_{TT}+2T^{-1}\phi_{T}=0\right\}.\end{equation*}
In this case the projection of $L$ on the plane of extensive variables has singularities.
\item singular submanifold $\Sigma_{2}\subset L$, where the differential form $dp\wedge dT$ degenerates:
\begin{equation*}\Sigma_{2}=\left\{\phi_{vv}=0\right\}.\end{equation*}
In this case the projection of $L$ on the plane of intensive variables has singularities.
\end{itemize}
Singularities of the second type are of special interest for us. Let applicable domains on $L$ be separated by $\Sigma_{2}$ from the set of points where $\phi_{vv}>0$. This means that thermodynamic system has a number of \textit{phases}. A jump between two applicable points $a_{1}=(p,T,e_{1},v_{1})\in L$ and $a_{1}=(p,T,e_{2},v_{2})\in L$ corresponding to different phases, characterized by intensives $p$ and $T$ and the specific Gibbs free energy $\gamma=e-T\sigma+pv$ conservation law $\gamma(a_{1})=\gamma(a_{2})$ is what we call \textit{phase transition}.

One may show that the specific Gibbs free energy $\gamma$ can be expressed in terms of Massieu-Plank potential $\phi$ in the following way:
\begin{equation*}\gamma=RT(v\phi_{v}-\phi).\end{equation*}
From what follows, that points $(v_{1},T)$ and $(v_{2},T)$ can be found from equations:
\begin{equation}
\label{PhaseEqui}
\begin{split}
&\phi _{v}\left( v_{2},T\right)=\phi _{v}\left( v_{1},T\right),{}\\&
\phi \left( v_{2},T\right)- v_{2}\phi _{v}\left(
v_{2},T\right)= \phi \left( v_{1},T\right) -v_{1}\phi _{v}\left( v_{1},T\right).
\end{split}
\end{equation}
Equations (\ref{PhaseEqui}) allow to construct coexistence curves in $\mathbb{R}^{2}(v,T)$. Equivalent equations
\begin{equation}
\label{PhaseEqui1}
\begin{split}
&\phi _{v}\left( v_{2},T\right)=\frac{p}{RT},\quad\phi _{v}\left(
v_{1},T\right) =\frac{p}{RT},{}\\&
\phi \left( v_{2},T\right)- v_{2}\phi _{v}\left(
v_{2},T\right)= \phi \left( v_{1},T\right) -v_{1}\phi _{v}\left( v_{1},T\right),
\end{split}
\end{equation}
allow to get coexistence curves in $\mathbb{R}^{3}(p,v,T)$ and in $\mathbb{R}^{2}(p,T)$. Coexistence curves show where phase transition occurs on the Lagrangian surface $L$.

\section{Peng-Robinson gases}
Peng-Robinson state equation was proposed by D.Y. Peng and D. Robinson in \cite{PR}. It appeared to be a superior description of nonpolar materials and became of wide use in petroleum industry. The first state equation (\textit{thermic} state equation) is of the form:
\begin{equation}\label{therm}f(p,T,e,v)=p-\frac{RT}{v-b}+\frac{a}{(v+b)^{2}-2b^2},\end{equation}
where $a$ and $b$ are constants responsible for the interaction between particles and particles' volume respectively. To define the Lagrangian surface $L$ completely, we need one more equation, the \textit{caloric} state equation. It can be obtained by means of (\ref{compat}). Assuming $g(p,T,e,v)=e-B(v,T)$ and taking the restriction of the Poisson bracket $[f,g]|_{L}$ we get an equation for $B(v,T)$:
\begin{equation*}B_{v}\left((v+b)^{2}-2b^{2}\right)-a=0,\end{equation*}
with solution of the form
\begin{equation*}B(v,T)=F(T)+\frac{a\sqrt{2}}{4b}\ln\left(\frac{v+b-\sqrt{2}b}{v+b+\sqrt{2}b}\right).\end{equation*}
Since taking $a=0$, $b=0$ in (\ref{therm}) we get an ideal gas state equation, we have to put $F(T)=nRT/2$, where $n$ is the degree of freedom.

Thus,
\begin{equation}\label{cal}g(p,T,e,v)=e-\frac{nRT}{2}-\frac{a\sqrt{2}}{4b}\ln\left(\frac{v+b-\sqrt{2}b}{v+b+\sqrt{2}b}\right),\end{equation}
and the Lagrangian surface for Peng-Robinson gases is given by (\ref{therm}) and (\ref{cal}). It is easy to check that the Massieu-Plank potential for Peng-Robinson gases is
\begin{equation}\label{MassPlank}\phi(v,T)=\ln\left(T^{n/2}(v-b)\right)-\frac{a\sqrt{2}}{4bRT}\ln\left(\frac{(3-2\sqrt{2})(v\sqrt{2}+v-b)}{v\sqrt{2}-v+b}\right).\end{equation}
Using (\ref{LegMani}) one can show that the specific entropy $\sigma$ for Peng-Robinson gases is (up to a constant)
\begin{equation}\label{ent}\sigma(v,T)=R\ln\left(T^{n/2}(v-b)\right),\end{equation}
and the Legendrian surface $\widehat L$ is defined by (\ref{therm}), (\ref{cal}) and (\ref{ent}).

By introducing the scale contact transformation
\begin{equation*}p\mapsto\frac{a}{b^{2}}p,\quad T\mapsto\frac{a}{bR}T,\quad e\mapsto\frac{a}{b}e,\quad v\mapsto bv,\quad\sigma\mapsto R\sigma\end{equation*}
one gets the reduced form of Peng-Robinson equations of state:
\begin{equation}\label{PRred}\begin{split}&p=\frac{T}{v-1}-\frac{1}{(v+1)^{2}-2},\quad e=\frac{nT}{2}+\frac{\sqrt{2}}{4}\ln\left(\frac{v\sqrt{2}+v-1}{v\sqrt{2}-v+1}\right),{}\\&\sigma=\ln\left(T^{n/2}(v-1)\right).\end{split}\end{equation}
Note that since $p>0$ and $T>0$, we consider only $v>1$.

The Massieu-Plank potential takes the form:
\begin{equation*}\phi(v,T)=\ln\left(T^{n/2}(v-1)\right)-\frac{\sqrt{2}}{4T}\ln\left(\frac{(3-2\sqrt{2})(v\sqrt{2}+v-1)}{v\sqrt{2}-v+1}\right).\end{equation*}

The differential quadratic form $\kappa$ can be written in the following way:
\begin{equation*}R^{-1}\kappa=-\frac{n}{2T^{2}}dT\cdot dT-\frac{Tv^{4}+2(2T-1)v^{3}+2(T+1)v^{2}-2(2T-1)v+T-2}{T(v-1)^{2}(v^{2}+2v-1)^{2}}dv\cdot dv.\end{equation*}
Since $n/2T^{2}>0$, the applicable domain for Peng-Robinson gases is given by inequality
\begin{equation*}T>\frac{2(v+1)(v-1)^{2}}{(v^{2}+2v-1)^{2}}\end{equation*}
and is shown in Figure \ref{aplicPR}.

\begin{figure}[h!]
\centering
\includegraphics[scale=.35]{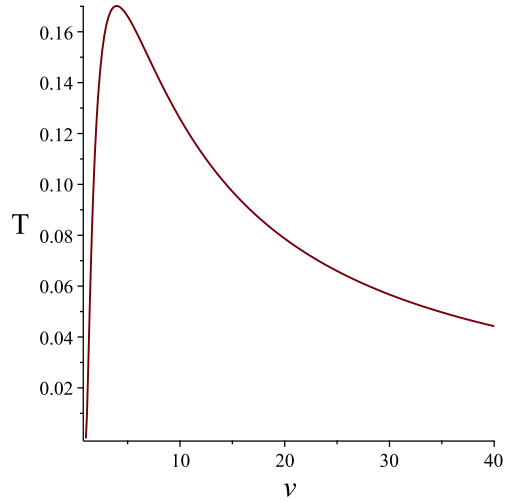}
\caption{Applicable domain for Peng-Robinson gases. The curve corresponds to the projection of singular submanifold $\Sigma_{2}\subset L$, where $\kappa$ changes its type. Applicable states are located above the curve.}
\label{aplicPR}       
\end{figure}

We can see that there is a \textit{critical} temperature $T_{0}$, and if $T>T_{0}$ there are no phase transitions.
\begin{theorem}
The critical temperature for Peng-Robinson gases $T_{0}$ and the corresponding critical volume $v_{0}$ are defined as follows:
\begin{equation*}v_{0}=1+2(4+2\sqrt{2})^{-1/3}+(4+2\sqrt{2})^{1/3},\quad T_{0}=\frac{2(v_{0}+1)(v_{0}-1)^{2}}{(v_{0}^{2}+2v_{0}-1)^{2}}.\end{equation*}
\end{theorem}

\subsection{Coexistence curves}
Coexistence curve is a level submanifold $\Gamma\subset L$ for the specific Gibbs potential $\gamma$. Such curves separate different phases of the medium and are given in terms of Massieu-Plank potential by (\ref{PhaseEqui}) and  (\ref{PhaseEqui1}). Resolving (\ref{PhaseEqui1}) with respect to $v_{1}$ and changing $v_{2}$ by $v$ we get $\Gamma\subset\mathbb{R}^{3}(p,v,T)$, which is shown in Figure \ref{coexPVT} and its projections $\Gamma_{1}\subset\mathbb{R}^{2}(p,T)$ and $\Gamma_{2}\subset\mathbb{R}^{2}(v,T)$, which are Figures \ref{coexPT} and \ref{coexVT} respectively.

\begin{figure}[h!]
\centering
\includegraphics[scale=.55]{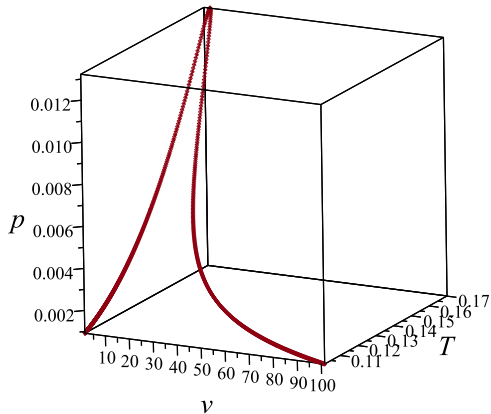}
\caption{Coexistence curve $\Gamma\subset\mathbb{R}^{3}(p,v,T)$ for Peng-Robinson gases.}
\label{coexPVT}       
\end{figure}

\begin{figure}[ht!]
\centering \subfigure[]{
\includegraphics[width=0.4\linewidth]{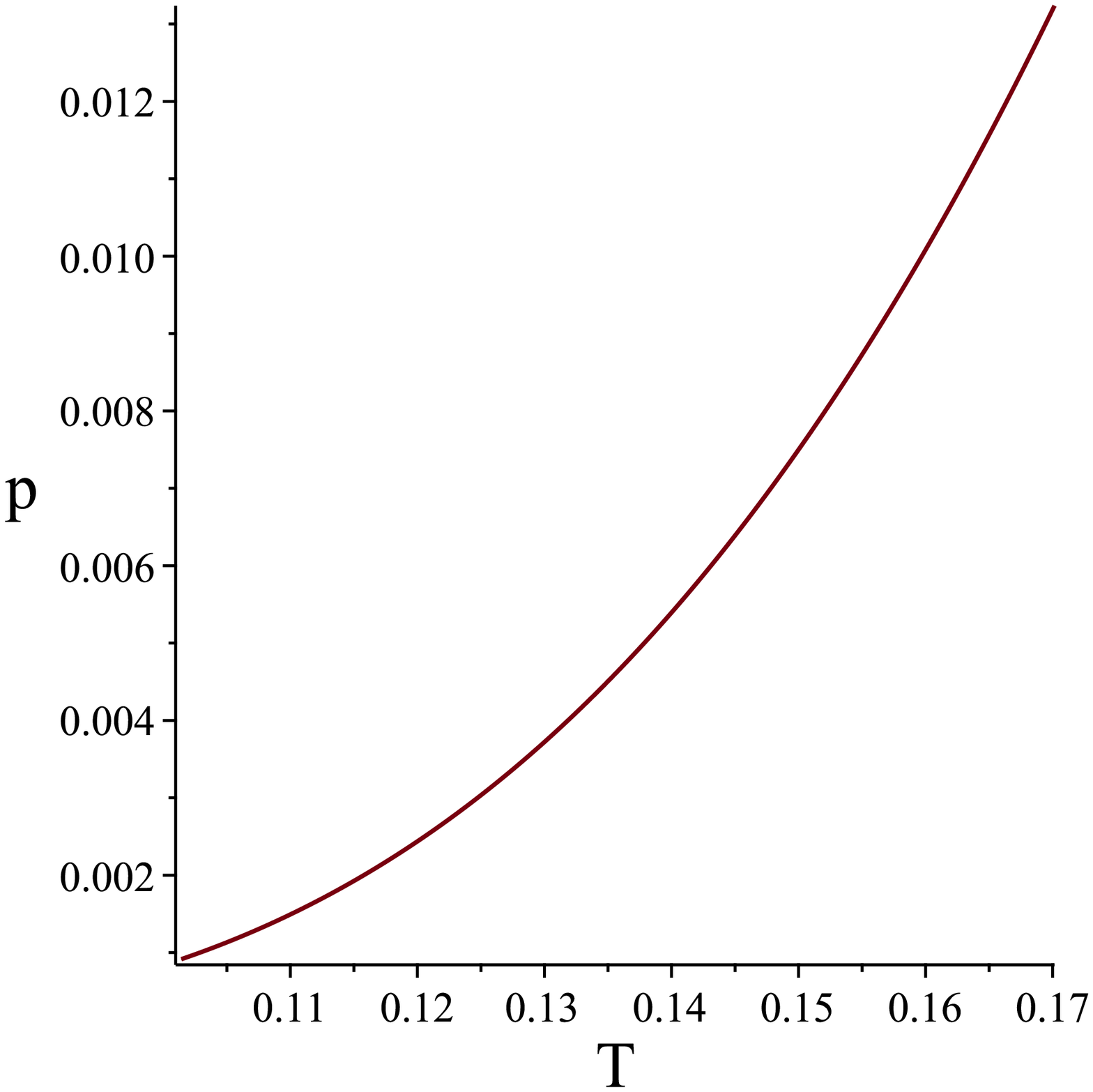} \label{coexPT} }
\hspace{4ex}
\subfigure[]{ \includegraphics[width=0.4\linewidth]{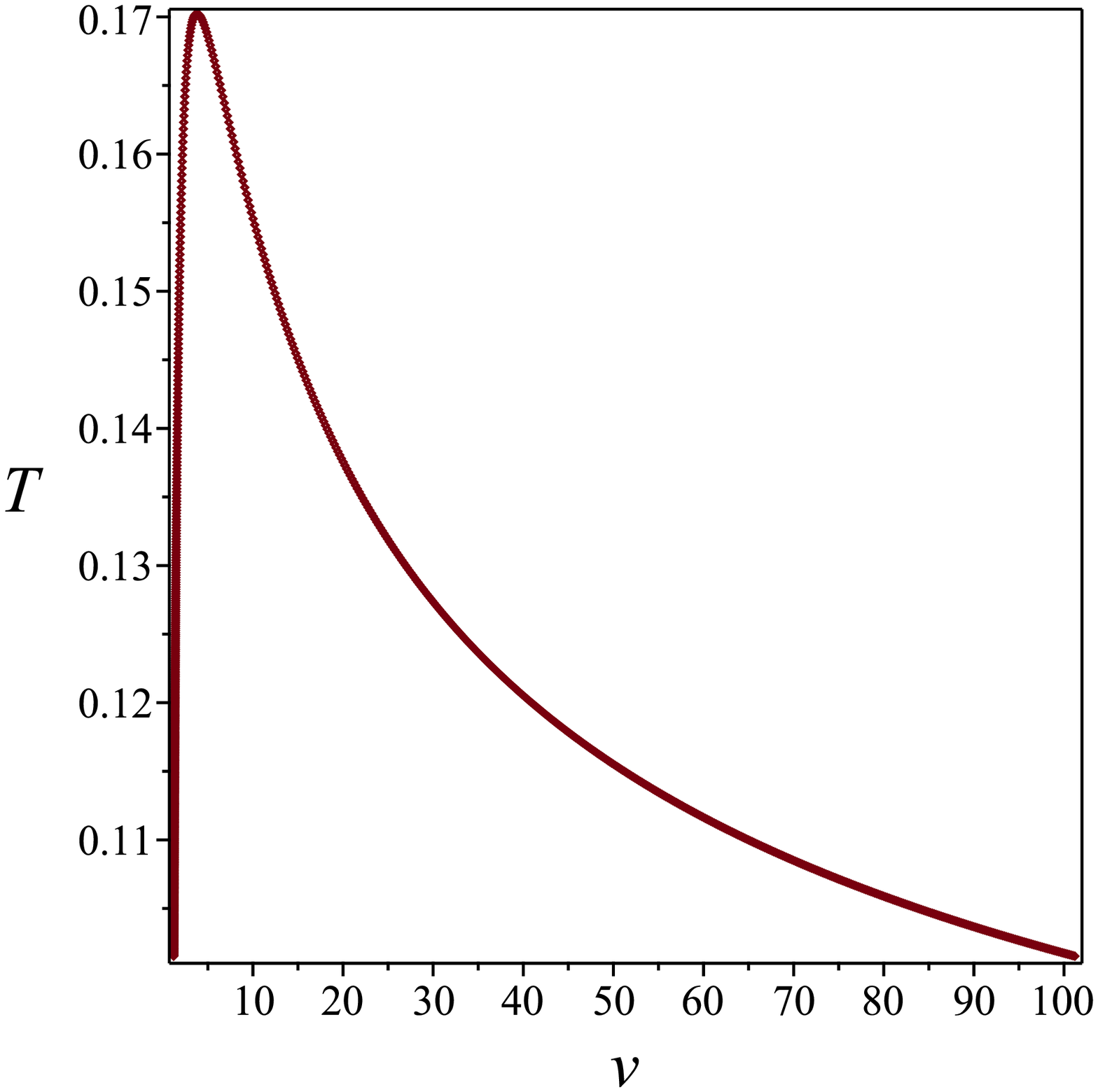} \label{coexVT} }
\caption{Coexistence curves for Peng-Robinson gases: \subref{coexPT} in $\mathbb{R}^{2}(p,T)$, liquid phase is on the left of the curve, gas phase is on the right; \subref{coexVT} in $\mathbb{R}^{2}(v,T)$, inside the curve --- intermediate state (condensation process).} \label{CoexCurves}
\end{figure}

\section{Solution of filtration equations}
Suppose that filtration domain $D\subset\mathbb{R}^{3}$ with a smooth boundary $\partial D$ contains a number $N$ of isotropic sources. If $N=1$, condition (\ref{adiab}) implies the constancy of the specific entropy, because stream lines intersect at the location of the source. Increasing the number of sources one gets one of two cases. Either stream lines intersect, which means that all the sources have to be of the same entropy, or not. The last means that sources may have different values of the specific entropy, but since stream lines do not intersect, filtrations with sources of different entropy are independent. Summarizing above discussion, we may say that the domain $D$ can be represented as a disjoint union of domains $D=\cup D_{k}$, where each of subsets $D_{k}$ contains sources with common specific entropy. Therefore, we restrict ourselves to consider adiabatic filtration with a given level of the specific entropy $\sigma_{0}$.

Fixed level of the specific entropy $\sigma_{0}$ and equations of state allow to express the temperature $T$ and the pressure $p$ as functions of the specific volume $v$. Namely, due to (\ref{LegMani}) we have:
\begin{equation}\label{entlev}\sigma_{0}=\phi+T\phi_{T}.\end{equation}
In applicable domain the derivative of the right-hand side of (\ref{entlev}) with respect to $T$ is positive, and $T(v)$ can be obtained as a root of (\ref{entlev}). Substituting the corresponding expression for $T(v)$ into the thermic equation of state, we get $p(v)$. General solution of (\ref{Darcy})-(\ref{adiab}) can be derived using the following theorem \cite{LRNon}:
\begin{theorem}
Equations of steady filtration (\ref{Darcy})-(\ref{adiab}) extended by equations of state are equivalent to the following equation:
\begin{equation*}\Delta(Q(v))=0,\end{equation*}
where
\begin{equation}\label{Qdef}Q(v)=-\int\frac{k}{v\mu}p^{\prime}(v)dv.\end{equation}
\end{theorem}
Thus, the Dirichlet problem for (\ref{Darcy})-(\ref{adiab}) has the following form:
\begin{equation}\label{Dir}\Delta(Q(v(x)))=0,\quad v|_{\partial D}=v_{0}.\end{equation}
Suppose that $N$ sources with given intensities $J_{i}$ are located at points $a_{i}\in D$, $i=\overline{1,N}$. According to the above theorem, a solution of (\ref{Dir}) (in general, multivalued) can be expressed explicitly by means of a harmonic in $D\setminus\left\{a_{i}\right\}$ function $u(x)$:
\begin{equation}\label{solsource}Q(v(x))=\sum\limits_{i=1}^{N}\frac{J_{i}}{4\pi|x-a_{i}|}+u(x),\quad u|_{\partial D}=Q(v_{0})-\sum\limits_{i=1}^{N}\left.\frac{J_{i}}{4\pi|x-a_{i}|}\right|_{\partial D}.\end{equation}

\section{Peng-Robinson gases filtration}
For Peng-Robinson gases expressions for $T(v)$ and $p(v)$ are following:
\begin{equation*}T(v)=s_{0}(v-1)^{-2/n},\quad p(v)=s_{0}(v-1)^{-1-2/n}-\frac{1}{(v+1)^{2}-2},\end{equation*}
where $s_{0}=\exp(2\sigma_{0}/n)$.

Due to (\ref{solsource}), a problem of invertibility of $Q(v)$ needs to be investigated. We need to find a specific entropy level $s_{0}$, such that $Q(v)$, including $s_{0}$ as a parameter, is invertible for any $v>1$.
\begin{theorem}
The function $Q(v)$ is invertible if the specific entropy constant $s_{0}$ satisfies the inequality:
\begin{equation*}s_{0}>\frac{2n(v_{0}+1)(v_{0}-1)^{2+2/n}}{(n+2)(v_{0}+2v_{0}-1)^{2}},\end{equation*}
where $v_{0}$ is the root of the equation:
\begin{equation}\label{pol}(2-n)v^{3}+3(n+2)v^{2}+(3n+2)v+3n-2=0.\end{equation}
There exists a real root of (\ref{pol}) $v_{0}>1$.
\end{theorem}
\begin{proof}
First of all, the invertibility condition for $Q(v)$ coincides with the condition of monotony. Function $Q(v)$ is monotonic if $Q^{\prime}(v)\ne0$ for any $v>1$. But due to (\ref{Qdef}) $Q^{\prime}(v)=0\Leftrightarrow p^{\prime}(v)=0$, which can be written as
\begin{equation*}\frac{s_{0}(n+2)}{2n}=w(v),\end{equation*}
where
\begin{equation*}w(v)=\frac{(v+1)(v-1)^{2/n+2}}{(v^{2}+2v-1)^{2}}.\end{equation*}
Therefore, if $s_{0}>2n(n+2)^{-1}\max\limits_{v>1} w(v)$, then $Q(v)$ is invertible. Condition $w^{\prime}(v)=0$ takes the form
\begin{equation*}P(v)=(2-n)v^{3}+3(n+2)v^{2}+(3n+2)v+3n-2=0,\end{equation*}
and the first part of the theorem is proved. Since the degree of freedom $n\ge3$, $P(+\infty)=-\infty$, $P(1)=8(n+1)>0$, from what follows the validity of the second statement of the theorem.
\end{proof}
Having an explicit solution given by (\ref{solsource}) and coexistence curves in the space of thermodynamic variables, shown above, we can construct submanifolds in $\mathbb{R}^{3}(x_{1},x_{2},x_{3})$ where phase transition occurs. The distribution of phases in space for $N=5$ sources located on a plane $x_{3}=0$ is presented in Figure \ref{phasespace}.

\begin{figure}[h!]
\centering
\includegraphics[scale=.35]{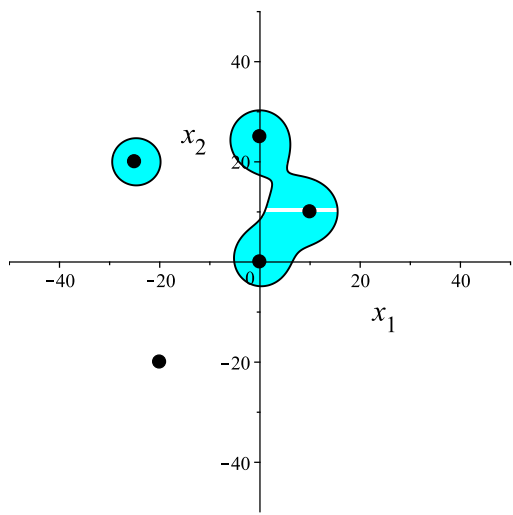}
\caption{The distribution of phases for Peng-Robinson gases. Coloured domain corresponds to the condensation process, black points are the sources.}
\label{phasespace}       
\end{figure}

The velocity field is shown in Figure \ref{velfield}.

\begin{figure}[h!]
\centering
\includegraphics[scale=.35]{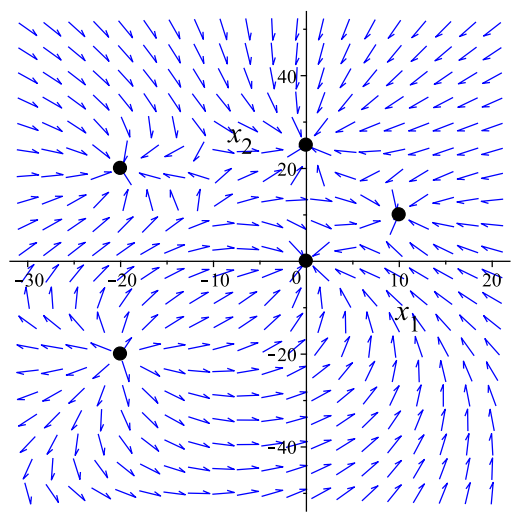}
\caption{The velocity field}
\label{velfield}       
\end{figure}

\par\bigskip\noindent
{\bf Acknowledgment.} This work was supported by the Russian Foundation for Basic Research (project 18-29-10013).

\bibliographystyle{amsplain}


\end{document}